\documentclass[submission, nomarks]{dmtcs-episciences}

\usepackage[utf8]{inputenc}
\usepackage{subfigure}
\usepackage{amsmath}
\usepackage[round]{natbib}
\newtheorem{theorem}{Theorem}
\newtheorem{lemma}{Lemma}
\newtheorem{definition}{Definition}
\newtheorem{corollary}{Corollary}
\author[Amanda Redlich]{Amanda Redlich\affiliationmark{1}\thanks{This material is based in part upon work supported by the National Science Foundation under Grant No. DMS-0931908 while the author was in residence at the Institute for Computational and Experimental Research in Mathematics in Providence, RI.}
}
\title{Power laws and power-of-two-choices}

\affiliation{University of Massachusetts Lowell, USA}

\begin{document}

\maketitle

\begin{abstract}
  This paper analyzes a variation on the well-known ``power of two choices" allocation algorithms.  Classically, the \emph{smallest} of $d$ randomly-chosen options is selected.  We investigate what happens when the \emph{largest} of $d$ randomly-chosen options is selected.  This process generates a power-law-like distribution: the $i^{th}$-smallest value scales with $i^{d-1}$, where $d$ is the number of randomly-chosen options, with high probability.  We give a formula for the expectation and show the distribution is concentrated around the expectation.  
\end{abstract}

\section{Introduction}
The basic framework of allocation algorithms is placing $m$ balls into $n$ bins, one at a time.  The idea of using the ``power of two choices" to develop a balanced allocation algorithm was introduced in \cite{azar}.  There, to place each ball, a random set of options (bins) was generated.  The ball was placed in the least-loaded option generated.

This original algorithm has given rise to a broad range of variations.  These include considering the heavily-loaded case (\cite{heavy}), the case with restrictions on feasible options (\cite{graphs, hypergraphs}), and selecting a dynamic number of options (\cite{us, us2}).  In each of these variations, the algorithm is designed to create a balanced allocation.

Here, consider a different variation, UNFAIR.  Instead of placing balls in the least-full bin, place balls in the most-full option.  Previous work \cite{me} introduced this algorithm and began its analysis.  It was shown that the algorithm does indeed generate a biased distribution, where a few bins contain most of the balls.

The main result of this paper is to show the distribution created by UNFAIR resembles a power law; the load of the $i^{th}$ bin scales with $i^{d-1}$.  Note that this is slightly different from the classical power laws, in which the \emph{probability} of a particular size $k$ scales with $k^c$.

Classical power laws are often studied in the context of preferential attachment models.  In preferential attachment, the probability of increase (``bin receiving a ball") is a function of the current size (``bin load").  In our case, the probability of increase depends on the \emph{relative} current size (``biggest bin load of the options") rather than its absolute value.  So it is not so surprising that the resulting distribution depends on the ranking of the bin, rather than its absolute size, as well.

The detailed analysis here shows that the number of balls in the $i^{th}$ lightest bin is $di^{d-1}(m/n^d)$ (plus a smaller error term) with high probability, for $m$ sufficiently large. The proof uses a phased analysis of the algorithm to show that after $m$ is sufficiently large, the behavior of the random allocation is predictable.

\begin{figure}
	\centering
	\includegraphics{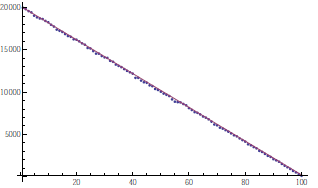}
	\caption{Experimental result for $n=100$, $m=10^6$, $d=2$}
	\label{fig:unbal2crop}
\end{figure}
\begin{figure}
	\centering
	\includegraphics{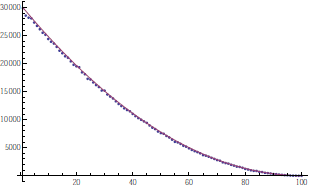}
	\caption{Experimental result for $n=100$, $m=10^6$, $d=3$}
	\label{fig:unbal3crop}
\end{figure}

\begin{figure}
	\centering
	\includegraphics{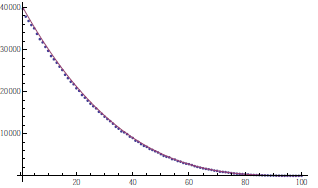}
	\caption{Experimental result for $n=100$, $m=10^6$, $d=4$}
	\label{fig:unbal4crop}
\end{figure}

Figures \ref{fig:unbal2crop}, \ref{fig:unbal3crop}, \ref{fig:unbal4crop} show experimental values for $d=2,3,4$.  The $y$-axis is bin loads, while the $x$-axis represents bins, sorted from most- to least-loaded. 
For ease of comparison, the predicted value is plotted as well. 
It is clear that the experimental data matches the theoretical bounds proved here.

\section{Definitions}
This paper concerns the allocation algorithm $\mathrm{UNFAIR}(m,n,d)$.  Informally, the algorithm places $m$ balls into $n$ bins by, at each time step, selecting $d$ options at random, and then placing the ball in the most-loaded bin.  In other words, it is a rich-get-richer process.  For clarity, the formal definitions and notation from \cite{me} are reiterated here.  

\begin{definition}
The unbalanced allocation process $\mathrm{UNFAIR}(m,n,d)$ is defined as follows. We use the standard notation of placing $m$ balls into $n$ bins using the power of $d$ choices.  At each time step, the process places a ball by selecting $d$ ``options" from the list of $n$ bins, with replacement; in other words, the process selects options uniformly at random from $\{1,\ldots n\}^d$. The option bins are all inspected.  The ball is placed in the currently \emph{most}-loaded option bin; in the case of a tie, the ball is placed uniformly at random in one of the most-loaded option bins.
\end{definition}

Initially, we think of bins as arbitrarily labeled $B_1, \ldots B_n$.  In analyzing the process, we reorder bins from least-loaded to most-loaded after each time step.  At time $t$, let $i(t)$ be the initial index of the $i^{th}$ lightest bin.  For ease of notation, let $B_{i(t)}$'s load at time $t$ be $b_{i(t)}$.  Break ties in this labeling randomly; for example, if $n=3$ and $b_{1}(t)=3$, $b_{2}(t)=2$, and $b_{3}(t)=3$, then $2(t)$ is equally likely to be 1 or 3.  We order bins $B_{1(t)}, B_{2(t)}, \ldots B_{n(t)}$, and have $b_{1(t)}(t)\leq b_{2(t)}(t)\leq \ldots \leq b_{n(t)}$.    

The implementation of $\mathrm{UNFAIR}$ under this labeling is the same as under the initial labeling.  An option set $(i_{1}, \ldots, i_{d})$ is chosen uniformly randomly from $[n]^{d}$ at time $t$. The bin $B_{i_{\mu}(t)}$ that gets the ball is such that $i_{\mu}=\max_{i_{j} \in S_{t} } \{ i_{j}\}$ and thus $b_{i_{\mu}(t)}(t)=\max_{i_{j} \in S_{t}} \{ b_{i_{j}(t)}(t) \}$.  (Note that ties are still broken uniformly randomly, as labels were applied to equally-loaded bins randomly.)

The reason for this reordering is that we now know that $i_{\mu}=\max \{i_{1}, \ldots, i_{d} \}$.   
In the initial labeling, it was equally likely that $b_{1}(t)>b_{2}(t)$ or $b_{2}(t)>b_{1}(t)$.  Now, it is always true that $b_{1(t)}(t) \leq b_{2(t)}(t)$.

\section{Results}

Because this is a rich-get-richer process, one would expect that, e.g., $b_{n(t)}(t)$ is significantly larger than $b_{1(t)}(t)$.  In other words, one would expect the bins to ``spread out" with time.  This is in fact the case, but the proof is complex: one has to deal with random tie-breaking, as well as quantify how much more likely a larger bin is to receive a ball.  

The main result is the following theorem, which gives an asymptotic value for the load of each bin under $\mathrm{UNFAIR}(m,n,d)$ for $m$ sufficiently large, with high probability.  In other words, we are able to quantify how much richer the rich get.  

\begin{theorem}\label{big}
After $m>n^{4d+13}$ balls have been placed under $\mathrm{UNFAIR}(m,n,d)$, the expected load $b_{i(m)}(m)$ in the $i^{th}$ smallest bin is $$\left((i/n)^d-((i-1)/n)^d\right)m+O\left(m/n^{d+1}\right).$$ Furthermore the probability that \emph{all} $n$ bins are within $O(m/n^{d+1})$ of their expectation is $1-O(n^{-d})$.
\end{theorem}

This theorem immediately leads to a corollary for indices linear in $n$.  In that case, the bins' loads take on a power-like distribution.  In particular, bins' loads are dependent on their relative ranking raised to the $d-1^{st}$ power.  The error bound proved here is $O(m/n^2)$; experimentally it seems this error is much smaller.
\begin{corollary}\label{power}
When $i=c n$ for constant $c$, after $m= \omega\left(n^{4d+13}\right)$ balls have been placed under $\mathrm{UNFAIR}(m,n,d)$, the expected load $b_{i(m)}(m)$ in the $i^{th}$ smallest bin is $$\left(dc^{d-1}\right)(m/n)+O(m/n^2).$$  Furthermore, the load is within $O(m/n^{d+1})$ of its expectation with probability $1-O(n^{-d})$.
\end{corollary}

Although Corollary \ref{power} is proved for large $m$ and $i$, it seems to hold experimentally for much smaller values.  For example, when $n=100$ and $m=1000000$, we already see a strong relationship, as in Figures \ref{fig:unbal2crop}, \ref{fig:unbal3crop}, \ref{fig:unbal4crop}.  

The proof of Theorem \ref{big} is involved; we first give a sketch, then go into detail.

\subsection{Proof sketch}
The proof involves an ``initialization" phase from time $0$ to time $t$, during which we prove the bins' behavior stabilizes.  We then analyze the post-initialization phase, from time $t+1$ to $m$, during which the bins behave predictably.

Without loss of generality, focus on a pair of bins, say $B_{i(t)}$ and $B_{j(t)}$.  We will first condition on the event $X'_{i,j}$, which is the event that at some time $s$ between $0$ and $t$, $|b_{i(t)}(s)-b_{j(t)}(s)|\geq g$ for some ``load gap" $g$.  Note that this event depends only on the sequence of option sets from time $0$ to time $t$; by conditioning on the event $X'_{i,j}$ we are simply restricting ourselves to a subset of the sequences of option sets from time $0$ to $t$.  The intersection $X'=\cap_{1\leq i<j\leq n} X'_{i,j}$ is a further restriction.  The first part of the proof is to show that almost all option sequences lie in this intersection.

Conditioned on the event $X'$, a gambler's ruin argument shows that with probability $1-ne^{-\delta}$, $B_{i(t)}$ will be the $i^{th}$ least heavily-loaded bin at any time in the future. In particular, then, with probability greater than or equal to $1-ne^{-\delta}$, $B_{i(t)}$ will be the $i^{th}$ most heavily-loaded bin at every time step from $t$ to $m$.  Call this event $Y_i$ and their intersection $Y=\cap_{i=1}^{n} Y_i$.  Again notice that the relative position of any bin depends only the sequence of option sets from time $0$ to time $m$; by conditioning on the events $X'$ and $Y$, we are again restricting ourselves to a subset of the option sequences.  The second part of the proof is to show that $1-\alpha$ proportion of option sequences from $0$ to $m$ lie in $X'Y$, where $\alpha$ is a small error term.

Once we condition on $X'$ and $Y$, we know that $B_i(t)$ is always the $i^{th}$ most heavily loaded bin from time $t$ to $m$.  Under uniform selection of option sets from $\{1,2, \ldots n\}^d$, the probability of the $i^{th}$ most heavily loaded bin receiving a ball would be $(i/n)^d-((i-1)/n^d$.  However, since we are conditioning on $X$ and $Y$, option sets are no longer uniform.  Fortunately, we know we are at most $\alpha$ from uniform.  Therefore we can bound from below by coupling the number of balls placed in $B_{i(t)}$ from time $t$ to $m$ conditioned on $X$ and $Y$ with the binomial distribution $$B (m-t, (i/n)^d-((i-1)/n^d-\alpha)$$ and from above by coupling with the binomial distribution $$B \left(m-t, \left((i/n)^d-((i-1)/n^d)\right)\left(\frac{1}{1-\alpha}\right)\right).$$  Using the standard binomial bound $Pr(|B(n, p)-np|>a)<2e^{-2a^2/n}$, we can see that the final load $b_{i(t)}(m)$ is $$(m-t)\left(\left(\frac{i}{n}\right)^d-\left(\frac{i-1}{n}\right)^d-\alpha\right)-\sqrt{nm}<b_{i(t)}(m)$$ and $$b_{i(t)}<(m-t)\left(\left(\frac{i}{n}\right)^d-\left(\frac{i-1}{n}\right)^d\right)\left(\frac{1}{1-\alpha}\right)+t+\sqrt{nm}$$ with probability $$1-\alpha-2e^{-2nm/(m-t)}=1-O(\alpha+e^{-2n}).$$  We show that $\alpha=O(n^{-d-1})$ to complete the proof.
\subsection{Initialization Lemmas}
We now go into detail about the initialization phase.  We first show that pairs of bins achieve load gaps with high probability, then show that the load gaps lead to bin ordering stabilization.

\begin{lemma}\label{init}
After $n^{3d+12}$ balls have been placed, with probability $1-O(n^{-d-1})$, every pair of bins has achieved a load gap of at least $n^2$.  
\end{lemma}
\begin{proof}
The initialization phase is actually a series of $\binom{n}{2}$ initialization phases, each of length $R=n^{3d+10}$.  The $i,j^{th}$ phase corresponds to the pair of bins ${B_i, B_j}$.  Without loss of generality, consider the $i,j^{th}$ phase $\phi_{i,j}$ focusing on bins $B_{i}$ and $B_{j}$.  Overall, there are $R$ balls placed in all bins during this phase.  

Let $X_{i,j}$ be the event that the number of balls placed in $B_{i}$ and $B_{j}$ combined is at least $r=n^{2d+10}$. Regardless of their loads, if the option set contains only bins $B_{i}$ and $B_{j}$, then one of them must receive a ball.  That happens with probability $(2/n)^d$.  So we may couple $\phi_{i,j}$ with $B(R, (2/n)^d))$, the binomial distribution of $R$ trials with success probability $(2/n)^d$, to get a lower bound on the number of balls placed in $B_{i}$ and $B_{j}$ during $\phi_{i,j}$.  A standard bound gives
\begin{align*}\label{E1k}
Pr(X_{i,j})&\geq 1-Pr(\beta \leq r)
\\ &\geq 1- 2\exp (-2(R(2/n)^d-r)^2/R))
\\ &\geq 1-O(e^{-n^{d+10}}).
\end{align*}
Let $X$ be the intersection of the $X_{i,j}$ over the entire initialization phase; in other words, $X$ is the event that every pair of bins $B_{i}$ and $B_{j}$ receives at least $r$ balls during its phase $\phi_{i,j}$.  Then the union bound gives
\begin{equation*}%\label{X}
Pr(X)\geq 1-\binom{n}{2}O(e^{-n^{d+10}})=1-O(n^2 e^{-n^{d+10}}).
\end{equation*}

We now condition on $X$ and for each $i,j$ consider the gap between $b_{i}$ and $b_{j}$ (i.e. the load gap of $B_{i}$ and $B_{j}$) during phase $\phi_{i,j}$.  We may model $|b_{i}-b_{j}|$ as a random walk.  Let $X'_{i,j}$ be the event that $|b_{i}-b_{j}|$ is $n^2$ at some time during $\phi_{i,j}$.  There is a bias towards the gap increasing, as a larger bin under this algorithm is more likely to receive a ball than a smaller bin.  Thus an $r$-step uniform random walk starting at $0$ coupled with $|b_{i}-b_{j}|$ would give a lower bound on this gap.  With this coupling, a standard bound from the theory of random walks (see e.g. \cite{feller}) gives
\begin{equation*}\label{X'}
Pr(X'_{i,j}|X) \geq 1-\frac{12n^2}{\sqrt{r}}=1-O(n^{-d-3}).
\end{equation*}

Let $X'$ be the intersection of the $X'_{i,j}$ over the entire initialization phase; in other words, $X$ is the event that every pair of bins $B_{i}$ and $B_{j}$ achieves a gap of at least $n^2$ during its phase.  Then the union bound gives
\begin{equation*}
Pr(X'|X)\geq 1-\binom{n}{2}\frac{12n^2}{\sqrt{r}}=1-O(n^{-d-1}).
\end{equation*}

Therefore we have the overall bound $$Pr(X') \geq Pr(X'|X)Pr(X)\geq (1-O(n^{-d-1}))(1-O(n^2e^{-n^{d+10}})\geq 1-O(n^{-d-1}),$$ as desired.
\end{proof}

We now show that  after the initialization phase the bin's relative positions are fixed with high probability from time $t$ to $m$: 

\begin{lemma}\label{fix}
For any $m>t=n^{3d+10}$, let $Y_i$ be the event that $i(t)=i(s)$ for all $t<s\leq m$ and $Y=\cap_{i=1}^{n} Y_i$.  Then $Pr(Y)\geq (1-O(n^{-d-1}))(1-O(n^2e^{-(d-1)n}))=1-O(n^{-d-1})$.
\end{lemma}

This lemma follows swiftly from Lemma \ref{init} combined with a theorem from \cite{me}:

\begin{theorem}[Theorem 4.3 from \cite{me}]\label{gambler}
	For any starting configuration of bins and balls, and for any $\delta$, if bin $B_{j}$ has a load of $ \delta n /(d-1)$ more balls than bin $B_{i}$'s load, then the probability of bin $B_{j}$'s load becoming smaller than bin $B_{i}$'s load under $\mathrm{UNFAIR}$ at any time in the future (i.e. after any number $m$ of balls has been added) is at most $e ^{-\delta}$.
\end{theorem}

\begin{proof}
We know from Lemma \ref{init} that the probability of $X'$, i.e. all bin pairs achieving a load gap of $n^2$ during the initialization phase, is at least $1-O(n^{-d-1})$.  We know from Theorem \ref{gambler} that if two bins achieve gap $n^2$ then the probability of swapping relative positions at any future time is $e^{-(d-1)n}$.  So in particular $$Pr(Y_i|X')\geq 1-ne^{-(d-1)n}.$$  The union bound gives $$Pr(Y|X')\geq 1-O(n^2e^{-(d-1)n}).$$ Putting this together with previous results gives $$Pr(Y) \geq Pr(YX')=Pr(Y|X')Pr(X')$$ $$Pr(Y)\geq (1-O(n^2e^{-(d-1)n}))(1-O(n^{-d-1}))\geq 1-O(n^{-d-1}),$$ as desired.
\end{proof}

\subsection{Proof of Theorem \ref{big}}

With these lemmas in hand, we are ready to prove Theorem \ref{big}.

\begin{proof}
Once we condition on $Y$, it seems relatively easy to estimate the load of $B_{i(t)}$: from time $t$ to $m$ $B_{i(t)}$ is always the $i^{th}$ most heavily loaded.  For a uniformly selected option set from $\{1,2, \ldots n\}^d$ the probability of bin $B_{i(t)}$ being the heaviest chosen is $(i/n)^d-((i-1)/n)^d$.  In other words, the number of balls added from time $t$ to $m$ should behave as a binomial distribution with $m-t$ trials and success probability $(i/n)^d-((i-1)/n)^d$.  

However, since we are conditioning on $Y$, we may no longer assume option sets are uniformly distributed in $\{1,2, \ldots n\}^d$.  Fortunately we know from Lemma \ref{fix} that we have removed at most $\alpha=O(n^{-d-1})$ of the option sets.  Therefore we may couple the number of balls placed in $B_{i(t)}$ from time $t+1$ to $m$ with binomial distributions to bound it from above and below: $$B(m-t, \left(\frac{i}{n}\right)^d-\left( \frac{i-1}{n}\right)^d-\alpha) \leq b_{i(m)}\leq t+B\left(m-t, \left(\left(\frac{i}{n}\right)^d-\left(\frac{i-1}{n}\right)^d\right)\left(\frac{1}{1-\alpha}\right)\right).$$  

Observe that, with $m=\Omega(n^{4d+13})$, $t=n^{3d+12}$, $\alpha=O(n^{-d-1})$, the expected values of the lower and upper bounds are both $$\left((i/n)^d-((i-1)/n)^d\right)m+O\left(m/n^{d+1}\right)$$ as desired. All that remains is to use the standard binomial bound $$Pr(|B(n,p)-np|>a)<2e^{-2a^2/n}$$ with $a=\sqrt{nm}$ to see that the probability of $b_{i(m)}$ being within $O(m/n^{d+1})$ of the binomial expectation is $1-O(\alpha+e^{-2n})=1-O(n^{-d-1})$.  Using the union bound one final time tells us the probability that \emph{all} bins $B_i$ have load $\left((i/n)^d-((i-1)/n)^d\right)m+O\left(m/n^{d+1}\right)$ is $1-O(n^{-d})$.

\end{proof}

\section{Conclusions and Future Work}

This paper and \cite{me} introduce a new allocation algorithm, UNFAIR.  Together they show that UNFAIR generates a new type of power-law distribution, where the load of each bins scales with the power of its relative position.

It has been shown that many naturally-occurring structures can be modeled by preferential attachment (e.g. \cite{BA}).  It would be of interest to discover if UNFAIR is a valid model for some family of naturally-occurring structures. Another direction of inquiry is the relation between UNFAIR and other mathematical process.  For example, Achlioptas random graph processes (see e.g. \cite{nogiant, giant, lutz}), in which a random graph is built by deciding between multiple ``option" edges to add at each stage, may have connections to this.

Another obvious direction of future work is to improve the threshold for convergence.  From experimental data, it seems possible that the distribution converges well below the $n^{4d+13}$ limit proved here.  New ideas for the initialization phase could improve this.

\begin{acknowledgements}
Thanks to Chen Avin, Zvi Lotker, and Eli Upfal for many useful discussions.  Particular thanks to Chen Avin for experimental data and the initial impetus to prove a power-law property.  Thanks also to anonymous reviewers for many helpful comments.
\end{acknowledgements}

\nocite{*}
\bibliographystyle{abbrvnat}

\bibliography{bibliography.bib}
\label{sec:biblio}

\end{document}